\pdfoutput=1
\documentclass[11pt]{article}

\newif\ifblind
\blindfalse

\newif\ifdraft
\draftfalse

\usepackage{rpmacros}
\usepackage[dvipsnames,svgnames,x11names,hyperref]{xcolor}
\usepackage{stmaryrd}
\usepackage{amsthm}
\usepackage{xfrac}
\usepackage{mathpazo}
\usepackage{enumitem}
\usepackage{fullpage}
\usepackage{thmtools}
\usepackage[T1]{fontenc}
\usepackage[colorlinks]{hyperref}
\hypersetup{
  linkcolor=[rgb]{0,0,0.4},
  citecolor=[rgb]{0, 0.4, 0},
  urlcolor=[rgb]{0.6, 0, 0},
  pdftitle={Unconditional Separation between LL and CLL},
  pdfauthor={Shubham Bhardwaj}
}
\usepackage{lipsum}
\usepackage{calrsfs,bbm}
\usepackage{multirow}
\usepackage{setspace}

\usepackage{microtype}

\usepackage[linesnumbered,vlined,ruled]{algorithm2e}
\usepackage{mathtools}
\DontPrintSemicolon
\SetKwComment{Comment}{$\triangleright$\ }{}

\usepackage{amsthm}
\usepackage{thmtools,thm-restate}
\usepackage[nameinlink,capitalise]{cleveref}

\numberwithin{equation}{section}
\declaretheoremstyle[bodyfont=\it,qed=\qedsymbol]{noproofstyle}

\declaretheorem[name=Observation,numbered=no]{observation*}

\declaretheorem[numberlike=equation]{theorem}

\declaretheorem[name=Theorem,numbered=no]{theorem*}

\declaretheorem[numberlike=equation]{lemma}
\declaretheorem[name=Lemma,numbered=no]{lemma*}

\declaretheorem[name=Corollary,numbered=no]{corollary*}

\declaretheorem[name=Proposition,numbered=no]{proposition*}

\declaretheorem[name=Claim,numbered=no]{claim*}

\declaretheorem[name=Conjecture,numbered=no]{conjecture*}

\declaretheorem[name=Question,numbered=no]{question*}

\declaretheoremstyle[bodyfont=\it,qed=$\lozenge$]{defstyle} 

\declaretheorem[numberlike=equation,style=defstyle]{definition}
\declaretheorem[unnumbered,name=Definition,style=defstyle]{definition*}

\declaretheorem[unnumbered,name=Example,style=defstyle]{example*}

\declaretheorem[unnumbered,name=Notation=defstyle]{notation*}

\declaretheorem[unnumbered,name=Construction,style=defstyle]{construction*}

\declaretheorem[unnumbered,name=Remark,style=defstyle]{remark*}

\renewcommand{\phi}{\varphi}
\renewcommand{\epsilon}{\varepsilon}

\newcommand{\SP}{\Sigma\Pi}
\newcommand{\SPsize}[1]{(\SP)^k\operatorname{-size}}

\usepackage{nth}
\usepackage{intcalc}
\usepackage{etoolbox}
\usepackage{xstring}

\usepackage{ifpdf}
\ifpdf
\else
\usepackage[quadpoints=false]{hypdvips}
\fi

\newcommand{\ECCCurl}[2]{http://eccc.hpi-web.de/report/\ifnumcomp{#1}{>}{93}{19}{20}#1/#2/}

\newcommand{\shortECCC}[2]{\texttt{\href{http://eccc.hpi-web.de/report/\ifnumcomp{#1}{>}{93}{19}{20}#1/#2/}{eccc:TR#1-#2}}}

\newcommand{\parseECCC}[1]{
\StrSubstitute{#1}{TR}{}[\tmpstring]%
\IfSubStr{\tmpstring}{/}{ 
\StrBefore{\tmpstring}{/}[\ecccyear]%
\StrBehind{\tmpstring}{/}[\ecccreport]%
}{
\StrBefore{\tmpstring}{-}[\ecccyear]%
\StrBehind{\tmpstring}{-}[\ecccreport]%
}%
\shortECCC{\ecccyear}{\ecccreport}}

\newcounter{todo}

\ifdraft
\newcommand{\SBnote}[1]{\refstepcounter{todo}\textcolor{BlueGreen}{\guillemotleft Shubham: #1 \guillemotright\addcontentsline{tod}{subsection}{[SB]~#1}}}
\else
\newcommand{\SBnote}[1]{}
\fi

\newcommand{\ignore}[1]{}

\makeatletter

\newcommand\listtodoname{List of todos}
\newcommand\listoftodos{%
  \section*{\listtodoname}\@starttoc{tod}}
\makeatother

\allowdisplaybreaks
\title{Unconditional Separation between \texorpdfstring{$\mathsf{LL}$}{LL} and \texorpdfstring{$\mathsf{CLL}$}{CLL}}

\ifblind
\else
\author{
     {Shubham Bhardwaj{\thanks{Tata Institute of Fundamental Research, Mumbai, India. Email: \texttt{shubham.bhardwaj@tifr.res.in}.}}}
}
\fi
\date{17 October 2026}

\begin{document}

\maketitle

\begin{abstract}
Catalytic computation involves the study of Turing machines that are provided an additional catalytic tape (with arbitrary content) that the machine may use while ensuring that the tape contents are restored at the end of the computation. Although it is known that catalytic log-space machines are at least as powerful as non-deterministic log-space (and more), we do not have unconditional asymptotic separations. 

In this work, we prove an unconditional separation between a deterministic space complexity class and its catalytic counterpart. More precisely, at the doubly logarithmic scale, we show that catalytic Turing machines can provably compute more languages than standard deterministic doubly-logarithmic space machines. The separation is witnessed by the unary language $\mathrm{POW2} = \{1^n : n \text{ is a power of two}\}$. We construct a catalytic Turing machine that uses $O(\log \log n)$ bits of space and $O(\log n)$ bits of catalytic tape to decide $\mathrm{POW2}$, whereas any standard deterministic Turing machine requires $\Omega(\log n)$ space to decide it.
\end{abstract}

\ifdraft
{\footnotesize 
\listoftodos
}
\fi

\section{Introduction}
\label{sec:introduction}

Catalytic computation is a model of computation in which a Turing machine is
provided with an additional \emph{catalytic tape} that is initially in an
arbitrary configuration. The machine may read from and write to the catalytic
tape during its computation, but it must restore the tape to its original
configuration before halting. This model was introduced by Buhrman, Cleve,
Kouck\'y, Loff, and Speelman~\cite{BuhrmanEtAl14}. We write $\mathsf{CL}$ for
the class of languages decided by catalytic machines with $O(\log n)$ ordinary
work space and polynomial catalytic space, and $\mathsf{CLP}$ for its
polynomial-time restriction.

Buhrman et al.~\cite{BuhrmanEtAl14} also gave the surprising result that evaluating a logarithmic-depth threshold circuit can be done in catalytic log-space, that is 
$\mathsf{TC}^1$ is contained in $\mathsf{CL}$. Since $\mathsf{TC}^1$ contains
both $\mathsf{NL}$ and $\mathsf{BPL}$, this gives
$\mathsf{NL}\subseteq\mathsf{CL}$ and $\mathsf{BPL}\subseteq\mathsf{CL}$, or alternatively, that directed $s$--$t$ reachability and estimation of random walks (complete problems for $\mathsf{NL}$ and $\mathsf{BPL}$, respectively) can be solved in catalytic log-space.

In particular, the conjecture $\mathsf{L}\neq\mathsf{NL}$ would imply the
strict containment $\mathsf{L}\subsetneq\mathsf{CL}$. This gave a strong indication that catalytic space may strictly provide more power than ordinary space bounded computation. Thus, they
gave strong, but conditional, evidence that catalytic log-space is more
powerful than standard log-space. On the other hand, they also showed that $\mathsf{CL}\subseteq\mathsf{ZPP}$. 

Recent work has also made progress to understand the connection between catalytic, non-deterministic, and randomized computation. Cook, Li, Mertz, and Pyne~\cite{CookLiMertzPyne25} proved the
unconditional derandomization $\mathsf{CBPL}=\mathsf{CL}$, where
$\mathsf{CBPL}$ is the bounded-error randomized analogue of $\mathsf{CL}$.
On the upper bound side, they also showed that $\mathsf{CL}\subseteq\mathsf{LOSSY}$, improving the previous upper bound $\mathsf{CL}\subseteq\mathsf{ZPP}$.
Kouck\'y, Mertz, Pyne, and Sami~\cite{KouckyMertzPyneSami25} subsequently
showed that nondeterminism also adds no power to catalytic log-space, proving
$\mathsf{CNL}=\mathsf{CL}$. Consequently, the deterministic,
nondeterministic, and bounded-error randomized versions of catalytic log-space
collapse:
\[
    \mathsf{CNL}=\mathsf{CBPL}=\mathsf{CL}.
\]
These results show that catalytic space is the ultimate resource in the log-space regime, subsuming both nondeterminism and bounded-error randomness, and adding nondeterminism or bounded-error randomness to catalytic log-space does not yield any additional computational power.

There has also been substantial progress in placing natural problems in
catalytic log-space. Agarwala and Mertz~\cite{AgarwalaMertz25} showed that
bipartite maximum matching can be solved in $\mathsf{CL}$. More recently,
general maximum matching was placed in $\mathsf{CL}$ by Chakraborty, Datta,
Kusre, Mukhopadhyay, and Sinhababu~\cite{ChakrabortyEtAl26}. Agarwala,
Alekseev, and Vinciguerra~\cite{AgarwalaAlekseevVinciguerra26} proved the more
general result that linear matroid intersection is in $\mathsf{CL}$. These
results show that $\mathsf{CL}$ contains central combinatorial
problems that are not known to be in $\mathsf{L}$, giving further evidence that catalytic log-space is more powerful than standard log-space.

Insights and techniques from catalytic computation also influenced the design of ordinary
space-efficient algorithms. Cook and Mertz~\cite{CookMertz20,CookMertz24} crucially used the ideas
from catalytic computation in their work on the Tree Evaluation problem, applying ideas from
catalytic space even outside the catalytic model. This line of work culminated
in their $O(\log n\log\log n)$-space algorithm for Tree
Evaluation~\cite{CookMertz24}, bringing this central candidate for separating
$\mathsf{L}$ from $\mathsf{P}$ close to logarithmic space.

The Cook--Mertz Tree Evaluation procedure was subsequently used by 
Williams~\cite{Williams25} to obtain a breakthrough time--space simulation.
By reducing the computation of a time-bounded multitape Turing machine to a
sequence of implicitly defined Tree Evaluation instances, Williams proved that,
for every $t(n)\geq n$,
\[
    \mathsf{TIME}(t(n))
    \subseteq
    \mathsf{SPACE}\!\left(O\!\left(\sqrt{t(n)\log t(n)}\right)\right).
\]
This substantially improves the previously known general simulation of time
$t$ in space $O(t/\log t)$.

Despite these advances, unconditional separations between
catalytic and ordinary space bounded classes are not known at any level. The containments of $\mathsf{NL}$ and
$\mathsf{BPL}$ in $\mathsf{CL}$ give a separation only under unproved
standard complexity assumptions, while the matching and matroid-intersection
results are not accompanied by log-space lower bounds. In this work, we give an unconditional separation
between a deterministic space class and its catalytic counterpart. We obtain
this separation at the doubly logarithmic scale, witnessed by the unary
language $\mathrm{POW2}$.

\begin{theorem}[Main theorem]
\label{thm:main}
There is a deterministic catalytic Turing machine that decides $\mathrm{POW2}$ using
$O(\log\log n)$ ordinary work space and $O(\log n)$ catalytic space. On the other
hand, no standard deterministic Turing machine decides $\mathrm{POW2}$ using
$o(\log n)$ work space. Consequently,
\[
    \mathrm{POW2}\in\mathsf{CLL}\setminus\mathsf{LL}
    \qquad\text{and}\qquad
    \mathsf{LL}\subsetneq\mathsf{CLL}.
\]
\end{theorem}

The theorem follows from \cref{thm:pow2-in-cll} and \cref{thm:pow2-not-in-ll}, which are proved in \cref{sec:catalytic-algorithm} and \cref{sec:lower-bound}, respectively. The upper bound is obtained by constructing a catalytic Turing machine for $\mathrm{POW2}$ that uses $O(\log\log n)$ ordinary work space and $O(\log n)$ catalytic space. The lower bound is obtained by applying a theorem of Hartmanis and Berman~\cite{HartmanisBerman76} that shows that every infinite unary language that can be decided by a deterministic Turing machine using $o(\log n)$ work space contains an infinite regular sublanguage. We then show that $\mathrm{POW2}$ contains no infinite regular sublanguage, which implies the lower bound.

\section{Preliminaries}
\label{sec:preliminaries}

Throughout the paper, all logarithms are to base $2$. We write
\[
    \mathrm{POW2}=\{1^n : n=2^k\text{ for some }k\geq 0\}.
\]

We write $\mathsf{REG}$ for the class of regular languages. Recall that a unary
language $R\subseteq\{1\}^*$ is regular if and only if its set of accepted
lengths is eventually periodic. Equivalently, there exist integers $N\geq 0$
and $p\geq 1$ such that, for every $m\geq N$,
\[
    1^m\in R
    \quad\Longleftrightarrow\quad
    1^{m+p}\in R.
\]

We use the standard off-line Turing machine model. A deterministic Turing machine has
a two-way read-only input tape and a two-way read-write work tape. On input
$x\in\{0,1\}^n$, the input is placed between two endmarkers, which the input head cannot
cross. The work tape is initially blank.
The work space of a computation is the number of distinct work-tape cells visited during
the computation.  All machines in this paper halt on every
input.

\begin{definition}For a function $s\colon\mathbb{N}\to\mathbb{N}$, let
$\mathsf{DSPACE}(s)$ denote the class of languages decided by deterministic Turing
machines using $O(s(n))$ work space on inputs of length $n$. At the doubly logarithmic
scale, we define
\[
    \mathsf{LL}=\mathsf{DSPACE}(\log\log n).
\]

\end{definition}
\begin{definition}[{Buhrman et al.~\cite{BuhrmanEtAl14}}]
A deterministic catalytic Turing machine has, in addition to the input and ordinary
work tapes, a two-way read-write \emph{catalytic tape} that consists of $c(n)$ cells and is initially in an arbitrary configuration
\[
    \tau\in\{0,1\}^{c(n)}.
\]
The ordinary work tape remains initially blank and unmarked, exactly as in the standard
model above. Likewise, the catalytic tape is unmarked, so the machine is not told where it ends. A catalytic machine must satisfy the following two conditions for every
input $x$ and every initial catalytic configuration $\tau$:
\begin{enumerate}
    \item its accept or reject decision depends only on $x$ and not on $\tau$; and
    \item when the machine halts, the catalytic tape again contains exactly $\tau$.
\end{enumerate}
Thus, the machine may read and modify the catalytic tape during its computation, but it
must restore all of its initial contents before halting.

We say that such a machine uses ordinary work space $s(n)$ and catalytic space $c(n)$
if it visits at most $O(s(n))$ cells of the ordinary work tape and at most $O(c(n))$ cells of
the catalytic tape. We denote the corresponding class of languages by
$\mathsf{CSPACE}(s,c)$.

The catalytic analogue of $\mathsf{LL}$ considered in this paper is
\[
    \mathsf{CLL}
    =\mathsf{CSPACE}(\log\log n,\log n).
\]
\end{definition}

\begin{remark*}
It only makes sense to consider catalytic tapes of length at most $2^{O(s(n))}$, i.e.,
$c(n)\leq 2^{O(s(n))}$, since specifying an arbitrary position on a catalytic tape of length $c(n)$
requires $\log c(n)$ bits, which fits within the ordinary work space precisely
in this regime.
\end{remark*}

Notice that $\mathsf{LL}\subseteq\mathsf{CLL}$, since a catalytic machine may simply
ignore its catalytic tape.

\section{A Catalytic Algorithm for \texorpdfstring{$\mathrm{POW2}$}{POW2}}
\label{sec:catalytic-algorithm}
We first show that $\mathrm{POW2}\in\mathsf{CLL}$ by constructing a catalytic Turing machine that decides $\mathrm{POW2}$ using $O(\log \log n)$ ordinary work space and $O(\log n)$ catalytic space. The machine uses the catalytic tape as a counter to store an arbitrary value $L$ and performs arithmetic modulo $2^{b+1}$ for various $b\leq \log n$. The machine first computes the largest power of two dividing $n$ and then checks whether $n$ is equal to that power of two. 

\begin{theorem}\label{thm:pow2-in-cll}
There is a deterministic catalytic Turing machine that decides
\[
    \mathrm{POW2} = \{1^n : n \text{ is a power of two}\}
\]
using $O(\log \log n)$ ordinary work space and $O(\log n)$ catalytic space. Moreover,
for every input and every initial content of the catalytic tape, the machine restores every
catalytic cell that it visits before halting.
\end{theorem}

\begin{proof}
We construct a deterministic catalytic Turing machine $M$ that decides $\mathrm{POW2}$ using $O(\log n)$ catalytic cells and $O(\log\log n)$ ordinary work space.

We think of the catalytic tape as a binary counter holding some arbitrary value $L$ and let
$\mathsf{enc}(L)$ be the encoding of this value that the catalytic tape holds. Let $\mathrm{Inc}(b)$ be the
routine that sweeps the input tape and increments the first $b+1$ bits of the catalytic counter by $1$ for each input
symbol it encounters, and similarly, let $\mathrm{Dec}(b)$ be the routine that sweeps the input tape and
decrements the first $b+1$ bits of the catalytic counter by $1$ for each input symbol it encounters. Clearly,
$\mathrm{Inc}(b)$ and $\mathrm{Dec}(b)$ are inverses of each other, and a single sweep of
$\mathrm{Inc}(b)$ followed by a single sweep of $\mathrm{Dec}(b)$ restores the catalytic tape to its
original content.

The machine has 2 phases. In the first phase, it finds the largest $v$ such that $2^v$ divides $n$. We index the catalytic bits starting from $0$, with bit $0$ being the least significant bit. To find $v$, the machine starts with $b=0$ and keeps increasing $b$ until the $b$-th bit detects a change. In each iteration, the machine stores the current value of the $b$-th bit of the catalytic tape. It then performs $\mathrm{Inc}(b)$ and checks the value of the $b$-th bit again. It then performs $\mathrm{Dec}(b)$ to restore the catalytic tape. If the two values are different, then it sets $v=b$ and stops the first phase; otherwise it increases $b$ by $1$ and continues.

\medskip

\SetKwComment{Comment}{$\triangleright$\ }{}
\begin{algorithm}[H]
\caption{Phase 1: finding the largest power of two dividing $n$}
\label{alg:phase-one}
\KwData{Input $1^n$, catalytic counter initially encoding an arbitrary value $L$}
\KwResult{A work-tape value $v=\nu_2(n)$; the catalytic tape is restored}
$b\gets 0$\;
\While{true}{
  $x\gets$ the $b$-th bit of the catalytic tape\;
  Perform $\mathrm{Inc}(b)$\Comment*[r]{adds $n$ to the counter modulo $2^{b+1}$}
  $y\gets$ the $b$-th bit of the catalytic tape\;
  Perform $\mathrm{Dec}(b)$\Comment*[r]{the counter is restored to $L$}
  \If{$x\neq y$}{
    $v\gets b$\;
    \Return $v$\;
  }
  $b\gets b+1$\;
}
\end{algorithm}

\medskip

After the first phase, the machine has computed $v=\nu_2(n)$, the largest power of two dividing
$n$. We now know that $n= m \cdot 2^v$ for some odd $m$. In the second phase, it checks whether $n=2^v$. To do this, we do $\mathrm{Inc}(v)$ and count the number of carries discarded at cell $v$ during the increment. If $n=2^v$ and the $v$-th bit of the catalytic tape is $0$, then there will be no carry discarded at cell $v$. If $n=2^v$ and the $v$-th bit of the catalytic tape is $1$, then there will be exactly one carry discarded at cell $v$. But if $m>1$, then the number of carries discarded at cell $v$ is larger than the $v$-th bit of the catalytic tape, since $n\geq 2^v+2^{v+1}$ and adding $n$ therefore causes at least one more carry than adding $2^v$. Thus, the machine accepts if and only if the number of carries discarded at cell $v$ is equal to the value of the $v$-th bit of the catalytic tape. The machine restores the catalytic tape by performing $\mathrm{Dec}(v)$ after counting the number of carries.

\medskip

\begin{algorithm}[H]
\caption{Phase 2: checking whether $n=2^v$}
\label{alg:phase-two}
\KwData{Input $1^n$, work-tape value $v=\nu_2(n)$, catalytic counter encoding $L$}
\KwResult{Accept iff $n=2^v$; the catalytic tape is restored}
$x\gets$ the $v$-th bit of the catalytic tape\;
Perform $\mathrm{Inc}(v)$\Comment*[r]{adds $n$ to the counter modulo $2^{v+1}$}
$y\gets$ the number of carries discarded at cell $v$ during $\mathrm{Inc}(v)$, counted up to $2$\;
Perform $\mathrm{Dec}(v)$\Comment*[r]{the counter is restored to $L$}
\If{$x\neq y$}{
  \Return reject\;
}
\Return accept\;
\end{algorithm}

\medskip

We now prove correctness. We first analyze Phase 1. The basic observation is the following.
For any $r$, the first $r$ bits of the catalytic counter encode the value of the full
counter modulo $2^r$. Thus, after one execution of $\mathrm{Inc}(r-1)$, these first $r$ bits change
from
\[
    L \bmod 2^r
    \quad\text{to}\quad
    (L+n) \bmod 2^r.
\]
In particular, if $2^r$ divides $n$, then the first $r$ bits do not change.

Now consider the iteration of Phase 1 with index $b$. If $b<\nu_2(n)$, then
$2^{b+1}$ divides $n$. By the observation with $r=b+1$, bits $0,1,\ldots,b$ are unchanged
after $\mathrm{Inc}(b)$, and hence bit $b$ is unchanged. Therefore Phase 1 cannot stop before
$b=\nu_2(n)$. On the other hand, when $b=\nu_2(n)$, we have
\[
    n\equiv 2^b \pmod {2^{b+1}}.
\]
Adding $n$ therefore flips bit $b$. Thus, we terminate with $b=\nu_2(n)$.

We next analyze Phase 2. The correctness follows from the observation that the number of carries discarded at cell $v$ during $\mathrm{Inc}(v)$ will be $x$, where $x$ is the value of the $v$-th bit of the catalytic tape, if and only if $n=2^v$. 

Finally, restoration is built into the two phases. Every execution of $\mathrm{Inc}(b)$ is followed
by the matching execution of $\mathrm{Dec}(b)$ before the machine uses the result of that
iteration to halt or continue. Hence every visited catalytic cell is restored for every initial
catalytic content.

We now show that the routines $\mathrm{Inc}(b)$ and $\mathrm{Dec}(b)$ can be
implemented within $O(\log\log n)$ space. Between routine calls, the catalytic
head is kept on cell $0$.

\begin{description}
\item[$\mathrm{Inc}(b)$ and $\mathrm{Dec}(b)$.] To perform $\mathrm{Inc}(b)$,
the machine sweeps its input head once across the input. Each time it reads a $1$, it keeps the input head in place and adds
$1$ to the catalytic counter by a carry forward increment: starting at cell
$0$, it replaces each $1$ by a $0$ and moves right, until it reads a $0$, which
it replaces by a $1$. If it replaces the $1$ in cell $b$ by a $0$, the carry out
of cell $b$ is discarded; this is exactly the reduction modulo $2^{b+1}$. In
Phase~2, the machine also counts the discarded carries, up to $2$. The machine
recognises cell $b$ by counting the moves of its catalytic head on a binary
counter on the work tape; the catalytic head then returns to cell $0$, and the
input sweep continues. The routine $\mathrm{Dec}(b)$ is identical with the
roles of $0$ and $1$ exchanged, so that a borrow propagates through the $0$'s
and is discarded at cell $b$. The binary counter never exceeds $b$, so both
routines use $O(\log b)$ work-tape cells and visit only cells $0,\dots,b$ of the
catalytic tape.

\item[Reading a catalytic bit.] To read bit $b$ (in Phase~1) or bit $v$ (in
Phase~2), the machine moves the catalytic head to that cell using a binary
counter in the same way, copies the bit to the work tape, and returns the
catalytic head to cell $0$. This uses $O(\log b)$ work-tape cells.
\end{description}

Altogether, the work tape holds the values $b$ and $v$, one binary counter for
the position of the catalytic head, and the values $x$ and $y$, each at most
$2$. Since $b\le v=\nu_2(n)\le\log n$, each of these values is stored in
$O(\log\log n)$ cells. Hence the ordinary work space is $O(\log\log n)$.
Moreover, the machine visits only cells $0,\dots,v$ of the catalytic tape, so it
uses at most $\log n+1$ catalytic cells and never needs to know where the
catalytic tape ends.

\end{proof}

\section{A Lower Bound for \texorpdfstring{$\mathrm{POW2}$}{POW2}}
\label{sec:lower-bound}

We now show the space lower bound for $\mathrm{POW2}$. We will use the following
theorem of Hartmanis and Berman.

\begin{theorem}[{Hartmanis and Berman~\cite{HartmanisBerman76}}]
\label{thm:unary-sublogarithmic-space}
Every infinite unary language that can be decided by a deterministic Turing machine
using $o(\log n)$ work space contains an infinite regular sublanguage.
\end{theorem}

We now show that $\mathrm{POW2}$ contains no infinite regular sublanguage.

\begin{lemma}
\label{lem:pow2-no-infinite-regular-subset}
The language $\mathrm{POW2}$ contains no infinite regular language.
\end{lemma}

\begin{proof}
Suppose that there is an infinite regular language $R\subseteq\mathrm{POW2}$. Since
$R$ is unary and regular, membership in $R$ is eventually periodic. Thus, there are
integers $N\geq 0$ and $p\geq 1$ such that, for every $m\geq N$,
\[
    1^m\in R
    \quad\Longleftrightarrow\quad
    1^{m+p}\in R.
\]
Since $R$ is infinite, we can choose $m\geq\max\{N,p+1\}$ for which $1^m\in R$.
The containment $R\subseteq\mathrm{POW2}$ implies that $m=2^k$ for some $k$.
Periodicity implies that $1^{m+p}\in R$, so $m+p$ must also be a power of two. On
the other hand, $p<m$, and therefore
\[
    m<m+p<2m.
\]
There is no power of two strictly between the consecutive powers of two $m$ and
$2m$, which gives a contradiction.
\end{proof}

From \cref{thm:unary-sublogarithmic-space} and \cref{lem:pow2-no-infinite-regular-subset}, we immediately obtain the following lower bound.

\begin{theorem}
\label{thm:pow2-not-in-ll}
No deterministic Turing machine with an initially blank and unmarked work tape can
decide $\mathrm{POW2}$ using $o(\log n)$ space. Consequently,
\[
    \mathrm{POW2}\notin\mathsf{LL}.
\]
\end{theorem}

\section{Conclusion}
\label{sec:conclusion}

We have given an unconditional separation between a deterministic space class
and its catalytic counterpart. The separation is witnessed by the unary
language $\mathrm{POW2}$. We constructed a catalytic Turing
machine that decides $\mathrm{POW2}$ using $O(\log\log n)$ ordinary work space
and $O(\log n)$ catalytic space. On the other hand, the theorem of Hartmanis
and Berman, together with the fact that $\mathrm{POW2}$ has no infinite regular
sublanguage, shows that no standard deterministic Turing machine decides
$\mathrm{POW2}$ using $o(\log n)$ work space. Consequently,
\[
    \mathsf{LL}\subsetneq\mathsf{CLL}.
\]
Thus, catalytic space provides a genuine increase in computational power
when the ordinary work tape has only doubly logarithmic size.

There is also a natural question below the doubly logarithmic threshold. In the
ordinary space model, every language recognized using
$o(\log\log n)$ work space is regular~\cite{StearnsHartmanisLewis65}. It is
therefore interesting to ask whether catalytic space can make useful
computation possible below this classical threshold: does there exist a
nonregular language that can be decided by a catalytic Turing machine using
$o(\log\log n)$ ordinary work space and $O(\log n)$ catalytic space? 

Another direction is to understand the role of nondeterminism at the doubly
logarithmic scale. The containment $\mathsf{NL}\subseteq\mathsf{CL}$ shows that
at logarithmic space, catalytic memory can simulate nondeterministic log-space.
The corresponding question for the classes considered in this paper is open.
Let $\mathsf{NLL}$ denote the class of languages decided by nondeterministic Turing machines using $O(\log\log n)$ work space.
It is natural to ask whether
\[
    \mathsf{NLL}\stackrel{?}{\subseteq}\mathsf{CLL}.
\]

The most prominent open question remains the analogous separation at the
logarithmic scale. In particular, proving
\[
    \mathsf{L}\subsetneq\mathsf{CL}
\]
unconditionally remains a major open problem. Even conditional separations or collapses based on weaker hardness assumptions would be of interest.

\section*{Acknowledgements}

We thank Andr\'as Salamon for pointing out an issue with the endmarkers of the catalytic tape in an
earlier version of this paper.

\bibliographystyle{customurlbst/alphaurlpp}
\bibliography{references}

\end{document}